%% file: anonymous-submission-latex-2023.tex
\documentclass[letterpaper]{article} 
\usepackage[submission]{aaai23}  
\usepackage{times}  
\usepackage{helvet}  
\usepackage{courier}  
\usepackage[hyphens]{url}  
\usepackage{graphicx} 
\usepackage{natbib}  
\usepackage{caption} 
\usepackage{algorithm}
\usepackage{algorithmic}

\usepackage{newfloat}
\usepackage{listings}

\input{preamble}

\DeclareCaptionStyle{ruled}{labelfont=normalfont,labelsep=colon,strut=off} 
\floatstyle{ruled}
\newfloat{listing}{tb}{lst}{}
\floatname{listing}{Listing}
\title{On the Complexity of Finding a Diverse and Representative Committee using a Monotone, Separable Positional Multiwinner Voting Rule}
\author {
    Kunal Relia\thanks{This work was supported, in part, by Julia Stoyanovich's NSF grants No. 1934464 and 1916505.}\\
    {\normalfont New York University, USA}\\
    {\normalfont krelia@nyu.edu}
}
\affiliations{
    Affiliation\\
    Affiliation Line 2\\
    name@example.com
}

\usepackage{bibentry}

\begin{document}

\maketitle

\begin{abstract}
\input{new_tex/abstract}
\end{abstract}

\input{new_tex/introduction}
\input{new_tex/related_work}
\input{new_tex/prelim}
\input{new_tex/complexity}

\input{new_tex/conclusion}

\section*{Acknowledgments}
I am grateful to Julia Stoyanovich for her insights. I thank Markus Brill, Benny Kimelfeld, and Phokion G. Kolaitis for many interesting and thought-provoking questions.

\bibliography{references_original}

\end{document}

%% file: preamble.tex
\usepackage{times}

\usepackage{soul}
\usepackage{url}
\usepackage[utf8]{inputenc}
\usepackage{graphicx}
\usepackage{amsmath}
\usepackage{booktabs}
\usepackage{subcaption}
\usepackage{xspace}
\usepackage{xcolor}
\usepackage{footnote}
\usepackage{amssymb}
\usepackage{pifont}

\newcommand{\csr}{committee selection rule\xspace}

\newcommand{\DiReCFp}[2]{DRCF\xspace}
\newcommand{\DiReCWD}{DRCWD\xspace}
\newcommand{\DiReCWDp}[2]{DRCWD\xspace}
\newcommand{\DiReCWDpse}[2]{DRCWD\xspace}
\newcommand{\DiReCWDpsu}[2]{DRCWD\xspace}

\newcommand{\group}[2]{A_{(#1, #2)}}
\newcommand{\groupsub}[2]{A_{(#1, #2_#1)}}
\newcommand{\population}[2]{P_{(#1, #2)}}
\newcommand{\populationsub}[2]{P_{(#1, #2_#1)}}

\newcommand{\etal}{\emph{et~al.}\xspace}
\usepackage{amsthm}
\newtheorem{definition}{Definition}
\newtheorem{theorem}{Theorem}

\newtheorem{claim}{Claim}
\usepackage{multirow}
\usepackage{xfrac}
\usepackage{subcaption}
\usepackage{balance}
\usepackage[switch]{lineno} %

\usepackage{imakeidx}
\usepackage{mathtools}

\DeclareMathOperator{\pos}{pos}

%% file: new_tex/abstract.tex

Fairness in multiwinner elections, a growing line of research in computational social choice, primarily concerns the use of constraints to ensure fairness. Recent work proposed a model to find a diverse \emph{and} representative committee and studied the model's computational aspects. However, the work gave complexity results under major assumptions on how the candidates and the voters are grouped. Here, we close this gap and classify the complexity of finding a diverse and representative committee using a monotone, separable positional multiwinner voting rule, conditioned \emph{only} on the assumption that P $\neq$ NP. 





%% file: new_tex/introduction.tex
\section{Introduction}
\label{sec:intro}

Fairness has recently received particular attention from the computer science research community. For context, the number of papers that contain the words ``fair'' or ``fairness'' in their titles and are published at top-tier computer science conferences like NeurIPS and AAAI  grew at an average of 38\% year on year since 2018. Moreover, the conference ACM FAccT, formerly known as FAT*, was established in 2018 to ``bring together researchers and practitioners interested in fairness, accountability, and transparency in socio-technical systems''. Similarly, there is a growing trend in the Computational Social Choice (COMSOC) community towards the use of ``fairness''\footnote{We do not consider the research on Fair Resource Allocation due to the specificity in the use of the word ``fair''.)} \cite{bredereck2018multiwinner,celis2017multiwinner,cheng2019group,hershkowitzdistrict,shrestha2019fairness}.

However, the term ``fairness'' is used in varying contexts. For example, Bredereck \etal \shortcite{bredereck2018multiwinner} and  Celis \etal \shortcite{celis2017multiwinner} call diversity of candidates in committee elections as fairness, while Cheng \etal \shortcite{cheng2019group} call representation of voters in committee elections as fairness. Such context-specific use of the term narrates an incomplete story. Hence, Relia \shortcite{relia2021dire} unified the framework using the DiRe Committee model that combines the use of diversity and representation constraints.

In line with the conceptual difference, the use of constraints leads to setups of multiwinner elections that are technically different. For instance, diversity constraint is a property of candidates and representation constraint is a property of voters. The use of these constraints are mathematically as different as the regularity and uniformity properties of hypergraphs\footnote{The mathematical differences between, say, the vertex cover problem on $d$-regular hypergraphs and $k$-uniform hypergraphs are well-known \cite{bansal2010inapproximability, feige2003vertex}}. Hence, it is important to mathematically delineate the two notions.

A starting step towards understanding the difference between diversity and representation is having a classification of complexity of using diversity and representation constraints in multiwinner elections. The classification of complexity, while technically interesting and important in itself, enables a detailed understanding of the nuances that delineate the two notions. Our main contributions are the complexity results that are based on a singular assumption that P$\neq$NP.

%% file: new_tex/related_work.tex
\section{Related Work}
\label{sec:RW}

\subsection{COMSOC and Classification of Complexity}
Computational Social Choice research has particularly focused on classifying the complexity of a known social choice problem. For instance, Konzcac and Lang \shortcite{konczak2005voting} introduced the problem of voting under partial information. This led to line of research that aimed to classify the complexity of the problem of possible winners (candidate is a winner in at least one completion) and necessary winners (candidate is a winner in all completions) over all pure positional scoring rules \cite{baumeister2012taking,betzler2010towards,chakraborty2021classifying,kenig2019complexity,xia2011determining}. 

\subsection{Multiwinner Elections, Fairness, and Complexity}


Our work primarily builds upon the literature on constrained multiwinner elections. Fairness from candidate's perspective is discussed via the use of diversity constraints over multiple attributes and its use is known to be NP-hard, which has led to approximation algorithms and matching hardness of approximation results by Bredereck \etal \shortcite{bredereck2018multiwinner} and Celis \etal \shortcite{celis2017multiwinner}. Additionally, goalbase score functions, which specify an arbitrary set of logic constraints and let the score capture the number of constraints satisfied, could be used to ensure diversity \cite{uckelman2009representing}. On the other hand, the study of fairness from voters' perspective pertains the use of representation constraints \cite{cheng2019group}. Finally, due to the hardness of using diversity constraints over multiple attributes in approval-based multiwinner elections \cite{brams1990constrained}, these have been formalized as integer linear programs (ILP) \cite{potthoff1990use}.

Overall, our work is at the intersection of the interest COMSOC researchers have on classifying the complexity and fairness in multiwinner elections. Specifically, our work is the closest to the work by Bredereck \etal \shortcite{bredereck2018multiwinner}, Celis \etal \shortcite{celis2017multiwinner}, Cheng \etal \shortcite{cheng2019group} and Relia \shortcite{relia2021dire} but we differ as we: (i) provide a complete classification of complexity of using finding a diverse and representative committee using a monotone, separable positional multiwinner voting rule, (ii) our NP-hardness results hold for \emph{all} integer values of attributes, and (iii) our NP-hardness results are conditioned \emph{only} on the assumption that P$\neq$NP. 

%% file: new_tex/prelim.tex
\section{Preliminaries and Notation}
\label{sec:prelim}


\paragraph{Multiwinner Elections.} Let $E = (C, V )$ be an election consisting of a candidate set $C = \{c_1,\dots,c_m\}$ and a voter set $V = \{v_1,\dots,v_n\}$, where each voter $v \in V$ has a preference list $\succ_{v}$ over $m$ candidates, ranking all of the candidates from the most to the least desired. $\pos_{v}(c)$ denotes the position of candidate $c \in C$ in the ranking of voter $v \in V$, where the most preferred candidate has position 1 and the least preferred has position $m$. 

Given an election $E = (C,V)$ and a positive integer $k \in [m]$ (for $k \in \mathbb{N}$, $[k] = \{1, \dots, k\}$), a multiwinner election 
selects a $k$-sized subset of candidates (or a committee) $W$ using a multiwinner voting rule $\mathtt{f}$ (discussed later) such that the score of the committee $\mathtt{f}(W)$ is the highest. 
We assume ties are broken using a pre-decided priority order.

\paragraph{Candidate Groups.} 
The candidates have $\mu$ attributes, $A_1 , . . . , A_\mu$, such that $\mu \in \mathbb{Z}$ and $\mu \geq 0$. Each attribute $A_i$, $\forall$ $i \in [\mu]$, partitions the candidates into $g_i \in [m]$ groups, $\group{i}{1} , . . . , \groupsub{i}{g} \subseteq C$. Formally, $\group{i}{j} \cap  \group{i}{j'} = \emptyset$, $\forall j,j' \in [g_i], j \ne j'$. For example, candidates may have race and gender attribute  ($\mu$ = 2) with disjoint groups per attribute, male and female ($g_1$ = 2) and Caucasian and African-American ($g_2$ = 2). Overall, the set $\mathcal{G}$ of \emph{all} such arbitrary and potentially non-disjoint groups is $\group{1}{1} , . . . , \groupsub{\mu}{g} \subseteq C$. 



\paragraph{Voter Populations.}
The voters have $\pi$ attributes, $A'_1 , . . . , A'_\pi$, such that $\pi \in \mathbb{Z}$ and $\pi \geq 0$. The voter attributes may be different from the candidate attributes. Each attribute $A'_i$, $\forall$ $i \in [\pi]$, partitions the voters into $p_i \in [n]$ populations, $\population{i}{1} , . . . , \populationsub{i}{p} \subseteq V$. Formally, $\population{i}{j} \cap  \population{i}{j'} = \emptyset$, $\forall j,j' \in [p_i], j \ne j'$. For example, voters may have state attribute  ($\pi$ = 1), which has populations California and Illinois ($p_1$ = 2). Overall, the set $\mathcal{P}$ of \emph{all} such predefined and potentially non-disjoint populations will be $\population{1}{1} , . . . , \populationsub{\pi}{p} \subseteq V$. 

Additionally, we are given $W_P$, the winning committee $\forall$ $P \in \mathcal{P}$. We limit the scope of $W_P$ to be a committee instead of a ranking of $k$ candidates because when a \csr such as CC rule is used to determine each population’s winning committee $W_P$, then a complete ranking of each population’s collective preferences is not possible. 


\paragraph{Multiwinner Voting Rules.} There are multiple types of multiwinner voting rules, also called committee selection rules. 
In this paper, we focus on committee selection rules $\mathtt{f}$ that are based on single-winner positional voting rules, and are
monotone and submodular ($\forall A \subseteq B, f(A) \leq f(B)$ and $f(B) \leq f(A) + f(B \setminus  A)$), and specifically separable. \cite{bredereck2018multiwinner,celis2017multiwinner}. 
Specifically, a special case of submodular functions are separable functions: 
score of a committee $W$ is the sum of the scores of individual candidates in the committee. Formally, $\mathtt{f}$ is separable if it is submodular and $\mathtt{f}(W) = \sum_{c \in W}^{}\mathtt{f}(c)$ \cite{bredereck2018multiwinner}. 
Monotone and separable selection rules are natural and are considered good when the goal of an election is to shortlist a set of individually excellent candidates:


\begin{definition}
\label{def-kborda} \textbf{$k$-Borda rule} The $k$-Borda rule
outputs committees of $k$ candidates with the highest Borda scores.
\end{definition}

%% file: new_tex/complexity.tex
\section{Classification of Complexity}
\label{sec:compclass}


We now study the computational complexity of \DiReCWD due to the presence of candidate attributes \emph{and} voter attributes. Specifically, we establish the NP-hardness of \DiReCWD for various settings of $\mu$, $\pi$, and $\mathtt{f}$ via reductions from a single well known NP-hard problem, namely, the vertex cover problem on 3-regular\footnote{A 3-regular graph stipulates that each vertex is connected to exactly three other vertices, each one with an edge, i.e., each vertex has a degree of 3. The VC problem on 3-regular graphs is NP-hard.}, 2-uniform\footnote{A 2-uniform hypergraph stipulates that each edge connects exactly two vertices.} hypergraphs \cite{garey1979computers, alimonti1997hardness}.

Note that the reductions are designed to conform to the real-world stipulations of candidate attributes such that (i) each candidate attribute $A_i,  \forall i \in [\mu]$,  \emph{partitions} all the $m$ candidates into two or more groups and (ii) either no two attributes partition the candidates in the same way or if they do, the lower bounds across groups of the two attributes are not the same. For stipulation (ii), note that if two attributes partition the candidates in the same way and if the lower bounds across groups of the two attributes are also the same, then mathematically they are identical attributes that can be combined into one attribute. We use the same stipulations for voter attributes.

Finally, all our results are based \emph{only} on the assumption that P$\neq$NP.

\begin{theorem}\label{lemma:DiReCWDrepdiv1}
If $\mu=1$, $\forall \pi \in \mathbb{Z} : \pi\geq1$, and   $\mathtt{f}$ is the monotone, separable function arising from an arbitrary single-winner positional scoring rule, then \DiReCWD 
is NP-hard.
\end{theorem}
\begin{proof}[Proof Sketch]
The proof extends from the reduction used in the proof of Theorem 5 in \cite{relia2021dire} (full-version). Specifically, we have $m+ (n \cdot m) + 1$ candidates such that $m+ (n \cdot m)$ from the previous reduction are in one group and we have a dummy candidate $a$ in the second group. The diversity constraints are set to 1 for both the groups. The voters, voter populations, and rankings are the same as before except candidate $a$ is ranked first by all the voters. The representation constraints are set to 2 for all voter populations. The committee size is set to $k+1$. It is easy to see that the proof of correctness follows the proof of correctness of Theorem 5 \cite{relia2021dire} (full-version) with an addition of always selecting the candidate $a$.
\end{proof}
\begin{theorem}\label{thm:DiReCWDdivrep21}
If $\mu=2$, $\forall \pi \in \mathbb{Z} : \pi\geq 1$, and   $\mathtt{f}$ is the monotone, separable function arising from an arbitrary single-winner positional scoring rule, then \DiReCWD 
is NP-hard.
\end{theorem}
\begin{proof}[Proof Sketch]
We now build upon the reduction used in the proof of Corollary~\ref{lemma:DiReCWDrepdiv1}. The only change in the reduction is the addition of the second candidate attribute. In addition two groups already present under one attribute, we create two more groups under the second attribute such that the first group contains the $m$ candidates from $C$ and the second group contains $(n\cdot m) + 1$ candidates ($(n\cdot m)$ dummy candidates $D$ and 1 dummy candidate $a$). The diversity constraints are set to 1 for both the groups. The voters, voter populations, rankings, and the representation constraints are the same as before. The committee size is again set to $k+1$. It is easy to see that the proof of correctness follows the proof of correctness of Corollary~\ref{lemma:DiReCWDrepdiv1}.
\end{proof}
\begin{theorem}\label{thm:DiReCWDdivrep31odd}
If $\forall \mu \in \mathbb{Z} : \mu\geq 3$ and $\mu$ is an odd number, $\forall \pi \in \mathbb{Z} : \pi\geq 1$, and   $\mathtt{f}$ is the monotone, separable function arising from an arbitrary single-winner positional scoring rule, then \DiReCWD 
is NP-hard. 
\end{theorem}

\begin{proof}
We reduce an instance of vertex cover (VC) problem to an instance of \DiReCWD.

\paragraph{\underline{PART I: Construction}}

\paragraph{Candidates:} We have one candidate $c_i$ for each vertex $x_i \in X$, and 
$2 \mu^2 m - 7 \mu m + 2 \mu m n + 2 m n + 3 m$ dummy candidates $d \in D$ where $m$ corresponds to the number of vertices in the graph $H$, $n$ corresponds to the number of edges in the graph $H$, and $\mu$ is a positive, odd integer (hint: the number of candidate attributes). Specifically, we divide the dummy candidates into two types of blocks:
\begin{itemize}
    \item Block type $B_1$ consists of $m\mu-3m$ blocks and each block consists of three sets of candidates:
    \begin{itemize}
        \item Set $T_1$ consists of single dummy candidate, $d_{i, 1}^{T_1} \in T_1$, $\forall$ $i \in [1,$ $m\mu-3m]$.
        \item Set $T_2$ consists of $\mu-1$ dummy candidates, $d_{i, j}^{T_2} \in T_2$, $\forall$ $i \in [1,$ $m\mu-3m]$, $j \in [1,$ $\mu-1]$. 
        \item Set $T_3$ consists of $\mu-1$ dummy candidates, $d_{i, j}^{T_3} \in T_3$, $\forall$ $i \in [1,$ $m\mu-3m]$, $j \in [1,$ $\mu-1]$. 
    \end{itemize}
    \item Block type $B_2$ consists of $2 m n$ blocks and each block consists of one set of candidates: 
    \begin{itemize}
        \item Set $T_4$ consists of $\mu+1$ dummy candidates, $d_{i, j}^{T_4} \in T_4$, $\forall$ $i \in [1,$ $2 m n]$, $j \in [1,$ $\mu+1]$. 
    \end{itemize}
\end{itemize}

Hence, there are $(m\mu-3m)\cdot(1+\mu-1+\mu-1)$ dummy candidates in blocks of type $B_1$ and $(2 m n) \cdot(\mu+1)$ dummy candidates in blocks of type $B_2$. This results in a total of $2 \mu^2 m - 7 \mu m + 3 m$ dummy candidates of type $B_1$ and $2\mu mn + 2mn$ dummy candidates of type $B_2$. Thus, $|D|$ =  $2 \mu^2 m - 7 \mu m + 2\mu mn + 2mn + 3 m$. Note that the types of blocks and sets discussed above are different and independent from the candidate groups (discussed later) that are used to enforce diversity constraints.

Overall, we set $A$ = \{$c_1, \dots, c_m$\} and the dummy candidate set $D$ = \{$d_1, \dots, d_{2 \mu^2 m - 7 \mu m + 2 \mu m n + 2 m n + 3 m}$\}. Hence, the candidate set $C$ = $A \cup D$ is of size $|C|=$ $2 \mu^2 m - 7 \mu m + 2 \mu m n + 2 m n + 4 m$ candidates. 

\paragraph{Committee Size:} We set the target committee size to be $k + m \mu^2 + 2mn\mu - 3m\mu $.

\paragraph{Candidate Groups:} We now divide the candidates in $C$ into groups such that each candidate is part of $\mu$ groups as there are $\mu$ candidate attributes. 

\underline{Candidates in $A$:}
Each edge $e \in E$ that connects vertices $x_i$ and $x_j$ correspond to a candidate group $G \in \mathcal{G}$ that contains two candidates $c_i$ and $c_j$. As our reduction proceeds from a 3-regular graph, each vertex is connected to three edges. This corresponds to each candidate $c \in A$ having three attributes and thus, belonging to three groups. 

Additionally, each candidate $c \in A$ is part of $\mu-3$ groups where each group is with the one candidate from Set $T_1$ of block type $B_1$. Specifically, candidate $c_i$ forms a group each with $d_{j, 1}^{T_1} \in T_1$ : $j\in[1+(i-1)(\mu-3),$ $i(\mu-3)]$. Hence, as each one of the $m$ candidates form $\mu-3$ groups, we have a total of $m(\mu-3)=m\mu-3m$ blocks of type $B_1$ consisting of $m(\mu-3)$ candidates in Set $T_1$. Overall, each candidate $c \in A$ has $\mu$ attributes and is part of $\mu$ groups.

\begin{table*}
\centering
\resizebox{\textwidth}{!}{\begin{tabular}{c|lllllllllllllll}
$v_1^{1}$, \dots, $v_1^{n}$&$ c_{i_1} \succ c_{j_1} $ & $\succ$ & $ U_{2}$ & $ \succ$&$ U_{3} $ & $\succ$&$ U_{4}^{v_1} $ & $\succ$&$ U_{5}^{v_1} $ & $\succ$&$ U_{6}^{v_1} $ & $\succ$&$ U_7$  \\[5pt]

$v_2^{1}$, \dots, $v_2^{n}$&$ c_{j_1} \succ c_{i_1} $ & $\succ$ & $ U_{2}$ & $ \succ$&$ U_{3} $ & $\succ$&$ U_{4}^{v_2} $ & $\succ$&$ U_{5}^{v_2} $ & $\succ$&$ U_{6}^{v_2} $ & $\succ$&$ U_7$  \\[5pt]

$v_3^{1}$, \dots, $v_3^{n}$&$ c_{i_2} \succ c_{j_2} $ & $\succ$ & $ U_{2}$ & $ \succ$&$ U_{3} $ & $\succ$&$ U_{4}^{v_3} $ & $\succ$&$ U_{5}^{v_3} $ & $\succ$&$ U_{6}^{v_3} $ & $\succ$&$ U_7$  \\[5pt]

$v_4^{1}$, \dots, $v_4^{n}$&$ c_{j_2} \succ c_{i_2} $ & $\succ$ & $ U_{2}$ & $ \succ$&$ U_{3} $ & $\succ$&$ U_{4}^{v_4} $ & $\succ$&$ U_{5}^{v_4} $ & $\succ$&$ U_{6}^{v_4} $ & $\succ$&$ U_7$  \\[5pt]

\vdots&&&&&&&\\

$v_{2n-1}^{1}$, \dots, $v_{2n-1}^{n}$&$ c_{i_n} \succ c_{j_n} $ & $\succ$ & $ U_{2}$ & $ \succ$&$ U_{3} $ & $\succ$&$ U_{4}^{v_{2n-1}} $ & $\succ$&$ U_{5}^{v_{2n-1}} $ & $\succ$&$ U_{6}^{v_{2n-1}} $ & $\succ$&$ U_7$  \\[5pt]

$v_{2n}^{1}$, \dots, $v_{2n}^{n}$&$ c_{j_n} \succ c_{i_n} $ & $\succ$ & $ U_{2}$ & $ \succ$&$ U_{3} $ & $\succ$&$ U_{4}^{v_{2n}} $ & $\succ$&$ U_{5}^{v_{2n}} $ & $\succ$&$ U_{6}^{v_{2n}} $ & $\succ$&$ U_7$  \\[5pt]

\end{tabular}}
\caption{An instance of preferences of voters created in the reduction for the proof of Theorem~\ref{thm:DiReCWDdivrep31odd}.}
\label{tab:voterRankingsDivRepConsodd}
\end{table*}

\underline{Candidates of Block type $B_1$:}
Each candidate in the block type $B_1$ has $\mu$ attributes and are grouped as follows:
\begin{itemize}
    \item Each dummy candidate $d_{j, 1}^{T_1} \in T_1$ : $j\in[1+(i-1)(\mu-3),$ $i(\mu-3)]$ is in the same group as candidate $c_i \in A$. Additionally, it is also in $\mu-1$ groups, individually with each of $\mu-1$ dummy candidates, $d_{j,o}^{T_2} \in T_2$, $\forall o \in [1,$ $\mu-1]$. Thus, the each dummy candidate $d_{j,1}^{T_1} \in T_1$ has $\mu$ attributes and is part of $\mu$ groups.
    \item For each dummy candidate $d_{j,i}^{T_2} \in T_2$ : $j\in[1+(i-1)(\mu-3),$ $i(\mu-3)]$ and $\forall i \in [1,$ $\mu-1]$, it is in the same group as $d_{j,1}^{T_1}$ as described in the previous point. It is also in $\mu-1$ groups, individually with each of $\mu-1$ dummy candidates, $d_{j,i}^{T_3} \in T_3$. Thus, each dummy candidate $d_{j,i}^{T_2} \in T_2$ has $\mu$ attributes and is part of $\mu$ groups.
    \item For each dummy candidate $d_{j,i}^{T_3} \in T_3$ :  $j\in[1+(i-1)(\mu-3),$ $i(\mu-3)]$ and $\forall i \in [1,$ $\mu-1]$, it is in $\mu-1$ groups, individually with each of $\mu-1$ dummy candidates, $d_{j,i}^{T_2} \in T_2$, as described in the previous point. Next, note that when $\mu$ is an odd number, $\mu-1$ is an even number, which means Set $T_3$ of each block has an even number of candidates. We randomly divide $\mu-1$ candidates into two partitions. Then, we create $\frac{\mu-1}{2}$ groups over one attribute where each group contains two candidates from Set $T_3$ such that one candidate is selected from each of the two partitions without replacement. Thus, each pair of groups is mutually disjoint. Thus, each dummy candidate $d_{j,i}^{T_3} \in T_3$ is part of exactly one group that is shared with exactly one another dummy candidate $d_{j',i}^{T_3} \in T_3$ where $j \neq j'$. Overall, this construction results in one attribute and one group for each dummy candidate $d_{j,i}^{T_3} \in T_3$. Hence, each dummy candidate $d_{j,i}^{T_3} \in T_3$ has $\mu$ attributes and is part of $\mu$ groups.
\end{itemize}

\underline{Candidates of Block type $B_2$:}
Finally, we assign candidates from the Block type $B_2$ to groups. Each of the dummy candidate in Set $T_4$ of each of the $2mn$ blocks is grouped individually with each of the remaining $\mu$ dummy candidates in Set $T_4$ of a block. Formally, $\forall i \in [1,$ $2mn]$, $\forall j \in [1,$ $\mu+1]$, $d_{i, j}^{T_4} \in T_4$ and $\forall o \in [1,$ $\mu+1]$, $d_{i, o}^{T_4} \in T_4$ such that $j\neq o$, $d_{i, j}^{T_4}$ and $d_{i, o}^{T_4}$ are in the same group. Hence, each block consists of $\mu+1$ candidates and each candidate is grouped pairwise with each of the remaining $\mu$ candidates. This means that each dummy candidate $d_{i, j}^{T_4} \in T_4$ has $\mu$ attributes and is part of $\mu$ groups\footnote{This setup of Block type $B_2$ is analogous to creating $2mn$ $K_{\mu+1}$ graphs, i.e., a total of $2mn$ complete (2-uniform, $\mu$-regular) (hyper-)graphs, each with $\mu+1$ vertices.}. 

\paragraph{Diversity Constraints:} We set the lower bound for each candidate group as follows:
\begin{itemize}
    \item $l^D_G=1$ for all $G \in \mathcal{G}$ : $G\cap A\neq \phi$, which corresponds that each vertex in the vertex cover should be covered by some chosen edge.
    \item $l^D_G=1$ for all $G \in \mathcal{G}$ such that at least one of the following holds:
    \begin{itemize}
        \item $\forall i \in [1,$ $m\mu-3m]$, $G \cap d_{i, 1}^{T_1} \neq \phi$.
        \item $\forall i \in [1,$ $m\mu-3m]$, $\forall j \in [1,$ $\mu-1]$, $G \cap d_{i, j}^{T_2} \neq \phi$.
        \item $\forall i \in [1,$ $m\mu-3m]$, $\forall j \in [1,$ $\mu-1]$, $G \cap d_{i, j}^{T_3} \neq \phi$.
    \end{itemize}
    \item $l^D_G=2$ for all $G \in \mathcal{G}$ such that $\forall i \in [1,$ $2mn]$, $\forall j \in [1,$ $\mu+1]$, $G \cap d_{i, j}^{T_4} \neq \phi$ and $\forall i \in [1,$ $2mn]$, $G\cap d_{i, \mu+1}^{T_4}=\phi$. 
    \item $l^D_G=1$ for all $G \in \mathcal{G}$ such that $\forall i \in [1,$ $2mn]$, $\forall j \in [1,$ $\mu+1]$, $G \cap d_{i, j}^{T_4} \neq \phi$.
\end{itemize}

In summary, $\forall$ $G \in \mathcal{G}$,

\[
    l^D_G = 
\begin{cases}
    2, & \text{if } \forall i \in [1, 2mn], \forall j \in [1, \mu], G \cap d_{i, j}^{T_4} \neq \phi \text{ and } \\
    & \forall i \in [1, 2mn], G\cap d_{i, \mu+1}^{T_4}=\phi\\
    1,              & \text{otherwise}
\end{cases}
\]

\paragraph{Voters and Preferences:} We now introduce $2n^2$ voters, $2n$ voters for each edge $e \in E$. More specifically, an edge $e \in E$ connects vertices $x_i$ and $x_j$. Then, the corresponding $2n^2$ voters $v \in V$ rank the candidates as follows:

\begin{itemize}
    \item \textbf{First $2$ positions} - Set $U_1$ : The top two positions are occupied by candidates $c_i$ and $c_j$ that correspond to vertices $x_i$ and $x_j$.  For voter $v^b_a$  where $a \in [2n]$ and $b \in [n]$, we denote the candidates $c_i$ and $c_j$ as $c_{i_{a}}$ and $c_{j_{a}}$. Note that we have two voters that correspond to each edge and hence, two voters $v^b_a$ and $v^b_{a-1}$ $\forall a \in [2n]_{\text{even}}$, rank the same two candidates in the top two positions. Specifically, voter $v^b_a$ ranks the two candidates based on the non-increasing order of their indices. Voter $v^b_{a-1}$ ranks the two candidates based on the non-decreasing order of their indices.
    \item \textbf{Next $m\mu^2-3m\mu$ positions} - Set $U_2$ : All of the $m\mu-3m$ dummy candidates from Set $T_1$ and all of the $m\mu^2-4m\mu+3m$ dummy candidates from Set $T_3$ are ranked based on non-decreasing order of their indices. Note that $m\mu-3m + m\mu^2-4m\mu+3m$ = $m\mu^2-3m\mu$.
    \item \textbf{Next $2mn\mu$ positions} - Set $U_3$ : $2mn\mu$ of the $2mn\mu+2mn$ dummy candidates from Set $T_4$ are ranked based on non-decreasing order of their indices. Specifically, a dummy candidate $d_{o, j}^{T_4}$ from Set $T_4$ is selected $\forall o \in [1,$ $2mn]$ and $\forall j \in [1,$ $\mu]$. Note that dummy candidate from Set $T_4$ of the kind $d_{o, \mu+1}^{T_4}$, for all $\forall o \in [1,$ $2mn]$, is not ranked in these positions. 
    \item \textbf{Next $m$ positions} - Set $U_4$ : $m$ out of $2mn$ unranked  dummy candidates from Set $T_4$ of the kind $d_{o, \mu+1}^{T_4}$, for all $\forall o \in [1,$ $2mn]$, are ranked in the next $m$ positions based on non-decreasing order of their indices. Specifically, the $m$ candidates that are ranked satisfy $(o\mod2n)+1=a$. Note that for each type of voter of the kind $v_i$, these $m$ candidates are distinct as shown below. Hence, for all pairs of voters of the kind $v_i, v_j \in V : v_i \neq v_j$, we know that $U_4^{v_i} \cap U_4^{v_j}=\phi$.
    \item \textbf{Next $m-2$ positions} - Set $U_5$ : The $m-2$ candidates from $A$, which are not ranked in the top two positions, are ranked based on non-decreasing order of their indices. Formally, $U_5=A \setminus \{c_{i_a}, c_{j_a}\}$.
    \item \textbf{Next $2mn-m$ positions} - Set $U_6$ : $2mn-m$ out of $2mn$ unranked dummy candidates from Set $T_4$ of the kind $d_{o, \mu+1}^{T_4}$, for all $\forall o \in [1,$ $2mn]$, are ranked based on non-decreasing order of their indices. Specifically, the candidates that are ranked satisfy $(o\mod2n)+1\neq a$.
    \item \textbf{Next $m\mu^2-4m\mu+3m$ positions} - Set $U_7$ : All of the $m\mu^2-4m\mu+3m$ dummy candidates from Set $T_2$ are ranked based on non-decreasing order of their indices.
\end{itemize}

More specifically, the voters rank the candidates as shown in Table~\ref{tab:voterRankingsDivRepConsodd}. The sets without a superscript (e.g., $U_2$) denote the candidate rankings  that are the same for all voters.

\paragraph{Voter Populations:} We now divide the voters in $V$ into populations such that each voter is part of $\pi$ populations as there are $\pi$ voter attributes. Specifically, the voters are divided into disjoint population over one or more attributes when $\forall \pi \in \mathbb{Z}, \pi\geq1$. The voters are divided into populations as follows: $\forall x \in [\pi]$, $\forall y \in [2n]$, $\forall z \in [n]$, voter $v_y^z \in V$ is part of a population $P\in\mathcal{P}$ such that $P$ contains all voters with the same ($z \mod x$) and $y$. 
Each voter is part of $\pi$ populations.

\paragraph{Representation Constraints:} We set the lower bound for each voter population as follows: $\forall P \in \mathcal{P}$, $l^R_P$ = $1+m\mu^2-3m\mu+2mn\mu$.

This completes our construction for the reduction, which is a polynomial time reduction in the size of $n$ and $m$. 

\paragraph{\underline{PART II: Proof of Correctness}}

\begin{claim}

We have a vertex cover $S$ of size at most $k$ that satisfies $e \cap S\neq\phi$ $\forall$ $e \in E$ if and only if we have a committee $W$ of size at most $k + m \mu^2 + 2mn\mu - 3m\mu $ such that $\forall$ $G \in \mathcal{G}$, $|G \cap W|\geq l^D_G$, $\forall$ $P \in \mathcal{P}$, $|W_P \cap W|\geq l^R_P$, and $\mathtt{f}(W) = \max_{W' \in \mathcal{W}} \mathtt{f}(W')$\footnote{The ties are broken based on rankings of the candidate such that higher ranked candidates are chosen over the lower ranked candidates. Note that with some trivial changes in voter preferences, we can ensure that there are no ties. Specifically, we create enough copies of each voter to ensure that each candidate not in the top $k + m \mu^2 + 2mn\mu - 3m\mu $ positions occupy the last position once.} where $\mathcal{W}$ is a set of committees that satisfy all constraints. 
\end{claim}

($\Rightarrow$) If the instance of the VC problem is a yes instance, then the corresponding instance of \DiReCWD is a yes instance.

\paragraph{Diversity constraints are satisfied:}
Firstly, each and every candidate group will have at least one of their members in the winning committee $W$, i.e., $|G \cap W|\geq l^D_G$ for all $G \in \mathcal{G}$.

More specifically, for each of the $m\mu-3m$ blocks of type $B_1$ of candidates, we select: 
\begin{itemize}
    \item one dummy candidate from Set $T_1$
    \item all $\mu-1$ dummy candidates from Set $T_3$
\end{itemize} 
This helps to satisfy the condition $l^D_G=1$ for all $G \in \mathcal{G}$ such that at least one of the following holds:
    \begin{itemize}
        \item $\forall i \in [1,$ $m\mu-3m]$, $G \cap d_{i, 1}^{T_1} \neq \phi$.
        \item $\forall i \in [1,$ $m\mu-3m]$, $\forall j \in [1,$ $\mu-1]$, $G \cap d_{i, j}^{T_2} \neq \phi$.
        \item $\forall i \in [1,$ $m\mu-3m]$, $\forall j \in [1,$ $\mu-1]$, $G \cap d_{i, j}^{T_3} \neq \phi$.
    \end{itemize}
Thus, we select $\mu$ candidates from $\mu-3$ blocks for each of the $m$ candidates that correspond to vertices in the vertex cover. This results in $(\mu \cdot (\mu-3) \cdot m) = m \mu^2 - 3m\mu$ candidates in the committee. 

Next, for each of the $2mn$ blocks of type $B_2$ of candidates, we select: 
\begin{itemize}
    \item $\mu$ dummy candidates $d_{i, j}^{T_4}$ from Set $T_4$ such that $i \in [1,$ $2 m n]$ and $j \in [1,$ $\mu]$.
\end{itemize} 
This helps to satisfy the conditions : $l^D_G=2$ for all $G \in \mathcal{G}$ such that $\forall i \in [1,$ $2mn]$, $\forall j \in [1,$ $\mu+1]$, $G \cap d_{i, j}^{T_4} \neq \phi$ and $\forall i \in [1,$ $2mn]$, $G\cap d_{i, \mu+1}^{T_4}=\phi$ and $l^D_G=1$ for all $G \in \mathcal{G}$ such that $\forall i \in [1,$ $2mn]$, $\forall j \in [1,$ $\mu+1]$, $G \cap d_{i, j}^{T_4} \neq \phi$. 
Overall, we select $\mu$ candidates from $2mn$ blocks. This results in additional $2mn\mu$ candidates in the committee. 

Finally, for groups that do not contain any dummy candidates, select $k$ candidates $c \in A$ that correspond to $k$ vertices $x \in X$ that form the vertex cover. These candidates satisfy the remainder of the constraints. Specifically, these $k$ candidates satisfy $|G \cap W|\geq 1$ for all the candidate groups that do not contain any dummy candidates. Hence, we have a committee of size $k + m \mu^2 + 2mn\mu - 3m\mu $.

\paragraph{Representation constraints are satisfied:}
Next, if the instance of the VC problem is a yes instance, then we have a winning committee $W$ of size $k + m \mu^2 + 2mn\mu - 3m\mu$ that consists of $k$ candidates corresponding to the VC and $m \mu^2 + 2mn\mu - 3m\mu$ candidates from Sets $U_2$ and $U_3$. Also, each and every population's winning committee, $W_P$ for all $P \in \mathcal{P}$, will have at least $1+m\mu^2-3m\mu+2mn\mu$ of their members in the winning committee $W$ such that $|W_P \cap W|\geq 1+m\mu^2-3m\mu+2mn\mu$, for all $P \in \mathcal{P}$, because:
\begin{itemize}
    \item as we have a yes instance of the VC problem, one of the two corresponding candidates occupying the first two positions of the ranking will be on the committee.
    \item each of the $m\mu^2-3m\mu$ candidates from Set $U_2$ will be on the committee
    \item each of the $2mn\mu$ candidates from Set $U_3$ will be on the committee
\end{itemize}

By construction, candidates in Set $U_2$ and candidates in Set $U_3$ will always be the part of each population's winning committee. Additionally, candidates in Set $U_2$ are the $\mu$ candidates selected from each of the $m\mu-3m$ blocks of the type $B_1$ and these candidates are the same across all voter populations. Also, the candidates in Set $U_3$ are the $\mu$ candidates selected from each of the $2mn$ blocks of the type $B_2$ and these candidates are the same across all voter populations. Thus, $|W_P \cap W|\geq 1+m\mu^2-3m\mu+2mn\mu$, for all $P \in \mathcal{P}$, and the \emph{same} winning committee $W$ satisfy the diversity constraints \emph{and} the representation constraints.

\paragraph{Highest scoring committee:}
It remains to be shown that $W$ is the highest scoring committee among all the committees that satisfy the given constraints.

Note that for a given population $P \in \mathcal{P}$, $\forall c \in C$ : $\pos_P(c)=1$ or $\pos_P(c)=2$, $\forall d_{i,j}^{T_4} \in $ Set $U_4^P$, $\forall v \in V$, $\pos_v(c)$ $\succ$ $\pos_v(d_{i,j}^{T_4})$. This holds based on the prerequisite condition that we are interested in committees that satisfy the constraints. Additionally, Set $U_2$ $\succ$ Set $U_3$ $\succ$ $d_{i,j}^{T_4}$, $\forall d_{i,j}^{T_4} \in $ Set $T_4$. This holds because $d_{i,j}^{T_4}$ is either in Set $U_4$ or Set $U_6$, and even after accounting for varying preferences of these two sets, Sets $U_2$ and $U_3$ are always ranked higher. Thus, the contribution of each candidate in Sets $U_2$ and $U_3$ to the total score will be greater than equal to a candidate $d_{i,j}^{T_4}$. As noted earlier, the ties are broken based on ranking. Hence, $W$ will be the highest scoring committee. 

Overall, a yes instance of the VC problem implies a yes instance of the \DiReCWD problem such that the committee $W$ s of size at most $k + m \mu^2 + 2mn\mu - 3m\mu $, $\forall$ $G \in \mathcal{G}$, $|G \cap W|\geq l^D_G$, $\forall$ $P \in \mathcal{P}$, $|W_P \cap W|\geq l^R_P$, and $\mathtt{f}(W) = \max_{W' \in \mathcal{W}} \mathtt{f}(W')$. 

($\Leftarrow$) 
The instance of the \DiReCWD is a yes instance when we have $k + m \mu^2 + 2mn\mu - 3m\mu $ candidates in the committee. This means that each and every group will have at least one of their members in the winning committee $W$, i.e., $|G \cap W|\geq l^D_G$ for all $G \in \mathcal{G}$. Then the corresponding instance of the VC problem is a yes instance as well. This is because the $k$ vertices $x \in X$ that form the vertex cover correspond to the $k$ candidates $c \in A$ that satisfy $|G \cap W|\geq 1$ for all the candidate groups that do not contain any dummy candidates. 

Next, by construction, we know that as a yes instance of the \DiReCWD problem satisfies the diversity constraints, each population's winning committee, $W_P$ for all $P \in \mathcal{P}$, will have at least $l^R_P$ of their members in the winning committee $W$, i.e., $|W_P \cap W|\geq l^R_P$ for all $P \in \mathcal{P}$. Importantly, as the $k + m \mu^2 + 2mn\mu - 3m\mu $-sized committee $W$ also satisfies the diversity constraints, we know that all candidates from Set $U_2$ and Set $U_3$ are selected. Additionally, no candidate from Set $U_4$ of each voter population is selected and \emph{at least one} of the top 2 ranked candidates from each voter population is selected. Hence, this means that the corresponding instance of the VC problem is a yes instance as well. Finally, this committee will be the highest scoring committee among all the committees that satisfy the constraints. This completes the other direction of the proof of correctness.
    
We note that as we are using a separable \csr, computing scores of candidates takes polynomial time. This completes the overall proof.

\end{proof}

The hardness holds even when either all diversity constraints are set to zero ($\forall G \in \mathcal{G}$, $l_G^D=0$) or all representation constraints are set to zero ($\forall P \in \mathcal{P}$, $l_P^R=0$).
Finally, we now show that \DiReCWD is NP-hard even when either all diversity constraints are set to zero ($\forall G \in \mathcal{G}$, $l_G^D=0$) \emph{or} all representation constraints are set to zero ($\forall P \in \mathcal{P}$, $l_P^R=0$). 

\begin{theorem}\label{thm:DiReCWDdivrep31even}
If $\forall \mu \in \mathbb{Z} : \mu\geq 3$ and $\mu$ is an even number, $\forall \pi \in \mathbb{Z} : \pi\geq 1$, and   $\mathtt{f}$ is the monotone, separable function arising from an arbitrary single-winner positional scoring rule, then \DiReCWD 
is NP-hard. 
\end{theorem}

\begin{proof}[Proof Sketch] The proof follows from the proof of Theorem~\ref{thm:DiReCWDdivrep31odd}. Hence, we only discuss the major changes in the construction of the proof.

\paragraph{\underline{PART I: Construction}}

\paragraph{Candidates:} We have two candidates $c_i$ and  $c_{m+i}$ for each vertex $x_i \in X$, and 
$2\cdot (2 \mu^2 m - 7 \mu m + 2 \mu m n + 2 m n + 3 m)$ dummy candidates $d \in D$ where $m$ corresponds to the number of vertices in the graph $H$, $n$ corresponds to the number of edges in the graph $H$, and $\mu$ is a positive, even integer (hint: the number of candidate attributes). We divide the dummy candidates into two types of blocks in line with the proof of Theorem~\ref{thm:DiReCWDdivrep31odd} but there twice the number of blocks and hence, by transitivity, twice the number of candidates.

\paragraph{Committee Size:} We set the target committee size to be $2\cdot(k + m \mu^2 + 2mn\mu - 3m\mu) $.

\paragraph{Candidate Groups:} We divide the candidates in $C$ into groups such that each candidate is part of $\mu$ groups as there are $\mu$ candidate attributes. The division is in line with the proof of Theorem~\ref{thm:DiReCWDdivrep31odd} but each set contains twice the number of candidates and there is one exception. For candidates in Block Type $B_2$:

\begin{itemize}
    \item For each dummy candidate $d_j^{T_3} \in T_3$, it is in $\mu-1$ groups, individually with each of $\mu-1$ dummy candidates, $d_i^{T_2} \in T_2$, as described in the previous point. Next, note that when $\mu$ is an even number, $\mu-1$ is an odd number, which means Set $T_3$ has an \emph{odd} number of candidates.  
    We randomly divide $\mu-2$ candidates into two partitions. Then, we create $\frac{\mu-2}{2}$ groups over one attribute where each group contains two candidates from Set $T_3$ such that one candidate is selected from each of the two partitions without replacement. Thus, each pair of groups is mutually disjoint. Hence, each dummy candidate $d_j^{T_3} \in T_3$ is part of exactly one group that is shared with exactly one another dummy candidate $d_{j'}^{T_3} \in T_3$ where $j \neq j'$. Overall, this construction results in one attribute and one group for all but one dummy candidate $d_j^{T_3} \in T_3$, which results into a total of $\mu$ attributes and $\mu$ groups for these $\mu-2$ candidates.
    This is because $\frac{\mu-2}{2}$ groups can hold $\mu-2$ candidates. Hence, one candidate still has $\mu-1$ attributes and is part of $\mu-1$ groups. If this block of dummy candidates is for candidate $c_i \in A$, then another corresponding block of dummy candidates for candidate $c_{m+i} \in A$ will also have one candidate $d_{z}^{T'_3} \in T'_3$ who will have $\mu-1$ attributes and is part of $\mu-1$ groups. We group these two candidates from separate blocks. Hence, now that one remaining candidate also has $\mu$ attributes and is part of $\mu$ groups. As there is always an even number of candidates in set $A$ ($|A|=2m$), such cross-block grouping of candidates among a total of $(\mu-3) \cdot 2m$ blocks, also an even number, is always possible.
\end{itemize}

\paragraph{Diversity Constraints:} We set the lower bound for each candidate group in line with the proof of Theorem~\ref{thm:DiReCWDdivrep31odd}.

\paragraph{Voters and Preferences:} We now introduce $4n^2$ voters, $4n$ voters for each edge $e \in E$. More specifically, an edge $e \in E$ connects vertices $x_i$ and $x_j$. Then, the corresponding $4n^2$ voters $v \in V$ rank the candidates in line with the proof of Theorem~\ref{thm:DiReCWDdivrep31odd} but each set now consists of twice the number of candidates. Specifically, the major change is in Set $U_1$ as it now consists of $c_i$, $c_j$, $c_{i+m}$, and $c_{j+m}$.

\paragraph{Voter Populations:} We divide the voters into populations in line with the proof of Theorem~\ref{thm:DiReCWDdivrep31odd}.

\paragraph{Representation Constraints:} We set the lower bound for each voter population as follows: $\forall P \in \mathcal{P}$, $l^R_P$ = $2\cdot(1+m\mu^2-3m\mu+2mn\mu)$.

This completes our construction for the reduction, which is a polynomial time reduction in the size of $n$ and $m$ in line with the proof of Theorem~\ref{thm:DiReCWDdivrep31odd}. 

\paragraph{\underline{PART II: Proof of Correctness}}

\begin{claim}

We have a vertex cover $S$ of size at most $k$ that satisfies $e \cap S\neq\phi$ $\forall$ $e \in E$ if and only if we have a committee $W$ of size at most $2\cdot(k + m \mu^2 + 2mn\mu - 3m\mu) $ such that $\forall$ $G \in \mathcal{G}$, $|G \cap W|\geq l^D_G$, $\forall$ $P \in \mathcal{P}$, $|W_P \cap W|\geq l^R_P$, and $\mathtt{f}(W) = \max_{W' \in \mathcal{W}} \mathtt{f}(W')$\footnote{The ties are broken based on rankings of the candidate such that higher ranked candidates are chosen over the lower ranked candidates.} where $\mathcal{W}$ is a set of committees that satisfy all constraints. 
\end{claim}

The proof of correctness easily follows from the proof of correctness of Theorem~\ref{thm:DiReCWDdivrep31odd}. 
\end{proof}

In line with the previous theorem, our reduction holds for a weaker result such that the hardness holds even when either all diversity constraints are set to zero ($\forall G \in \mathcal{G}$, $l_G^D=0$) or all representation constraints are set to zero ($\forall P \in \mathcal{P}$, $l_P^R=0$).

%% file: new_tex/conclusion.tex
\section{Conclusion}
\label{sec:conc}

We  classified the complexity of finding a diverse and representative committee using a monotone, separable positional multiwinner voting rule. In doing so, we established the close association between \DiReCWD and the vertex problem problem. This can lead to interesting future work about having an understanding of the existence of a diverse and representative outcome and complexity of finding one. Specifically, we can study if \DiReCWD under certain realistic assumptions is PPAD-complete or not based on knowledge about the vertex cover problem and its possible PPAD-completeness \cite{kiraly2013ppad}. Additionally, this association can be used to find realistic settings where the model becomes tractable and validate it by giving examples of real-world datasets where such algorithms may work.